\documentclass[letterpaper, 10 pt, conference]{ieeeconf}  

\IEEEoverridecommandlockouts                              
\usepackage{amsmath}
\usepackage{amssymb}

\newcommand{\set}[1]{\left\{#1\right\}}
\newcommand{\Real}{\mathbb R}
\newcommand{\eps}{\varepsilon}
\newcommand{\abs}[1]{\left\vert#1\right\vert}
\newcommand{\norm}[1]{\left\Vert#1\right\Vert}

\newtheorem{assumption}{Assumption}
\newcommand{\T}{\mathcal{T}}
\newcommand{\always}{\Box}
\newcommand{\eventually}{\Diamond}
\newcommand{\ltlx}{\text{LTL}_{\backslash\scriptstyle\bigcirc}}
\newcommand{\U}{\mathcal{U}}
\newcommand{\R}{\mathcal{R}}

\renewcommand{\S}{\mathcal{S}}

\renewcommand{\phi}{\varphi}
\renewcommand{\u}{\mathbf{u}}
\newcommand{\x}{\mathbf{x}}

\newcommand{\z}{\mathbf{z}}
\newcommand{\B}{\mathbb{B}}

\newcommand{\ra}{\rightarrow}

\newtheorem{remark}{Remark}
\newtheorem{ex}{Example}
\newtheorem{problem}{Problem}
\newtheorem{definition}{Definition}
\newtheorem{theorem}{Theorem}
\newtheorem{lemma}{Lemma}
\newtheorem{prop}{Proposition}

\usepackage{tikz}
\usepackage{tikz-cd}
\usetikzlibrary{shapes,arrows}

\title{\LARGE \bf
Robust Decidability of Sampled-Data Control of Nonlinear Systems with Temporal Logic Specifications
}

\author{Jun Liu
\thanks{Jun Liu is with the Department of Applied Mathematics, University of Waterloo, Waterloo, Ontario N2L 3G1, Canada. Email:         {\tt\small j.liu@uwaterloo.ca}}%
\thanks{This work was supported in part by the NSERC DG, CRC, and ERA programs.}}
\begin{document}

\maketitle
\thispagestyle{empty}
\pagestyle{empty}

\begin{abstract}
This paper explores the theoretical limits of using discrete abstractions for nonlinear control synthesis.  More specifically, we consider the problem of deciding continuous-time control with temporal logic specifications. We prove that sampled-data control of nonlinear systems with temporal logic specifications is robustly decidable in the sense that, given a continuous-time nonlinear control system and a temporal logic formula, one can algorithmically decide whether there exists a robust sampled-data control strategy to realize this specification when the right-hand side of the system is slightly perturbed by a small disturbance. If the answer is positive, one can then construct a (potentially less) robust sampled-data control strategy that realizes the same specification. The result is proved by constructing a  robustly complete abstraction of the original continuous-time control system using sufficiently small discretization parameters. We illustrate the result with three nonlinear control examples. 
\end{abstract}

\section{Introduction}

The control of dynamical systems to satisfy formal specifications (e.g. temporal logics) has received considerable attention in the past decade \cite{belta2017formal,tabuada2009verification}. 
This is partially motivated by the increasing demand of autonomous decision making by physical systems (e.g. mobile robots) in uncertain environments to achieve more complex tasks \cite{kloetzer2007temporal,fainekos2009temporal,kress2011correct}. Many system relations have been proposed as abstractions of nonlinear systems \cite{pola2008approximately,zamani2012symbolic,liu2016finite,reissig2017feedback,nilsson2017augmented}. Such abstractions are desirable for several reasons. First, they are sound in the sense that they can be used to design provably correct controllers with respect to a given formal specification. Second, they are often finite (e.g. finite transition systems) and the original control design problem over an infinite state space can be effectively solved as a search problem over a finite structure. Third, the construction of these abstractions can be automated   with the aid of a computer.

One of the main drawbacks of abstraction-based approaches is their computational cost, which is often incurred when a finer and finer abstraction is used in the hope of finding a controller when a coarser abstraction fails to yield one. However, without theoretical guarantees on completeness, i.e. 
if a control strategy exists, then it can be found by an abstraction-based approach, such computational efforts can be futile. This motivates the research in this paper. In this paper, we seek to answer the question whether a computational procedure exists to decide if a control strategy exists for a given formal specification. We consider general continuous-time nonlinear control systems, but restrict our attention to a specific class of control strategies, namely sample-and-hold control strategies. The main result of the paper shows that if there exists a robust sample-and-hold control strategy for the continuous-time nonlinear control system to realize a given temporal logic specification, then one can construct a robust control strategy for the system to realize the same specification. 

\subsection{Related work}

We review several results in the literature that are most relevant to the result presented in this paper. In \cite{tabuada2006linear}, it is shown that bisimilar (equivalent) symbolic models exist for controllable discrete-time linear systems and, as a result, temporal logic control for  discrete-time controllable linear systems is decidable. For nonlinear systems, the authors of \cite{pola2008approximately} showed that approximately bisimilar models can be constructed for incrementally stable systems \cite{angeli2002lyapunov}. The assumption of incremental stability essentially allows one to construct a deterministic transition system that can approximate a sample-data representation of the original nonlinear system to any degree of precision. For nonlinear systems without the incremental stability assumption, the authors of \cite{zamani2012symbolic} showed that symbolic models that approximately alternatingly simulate the sample-data representation of a general nonlinear control system can be constructed. Because a sampled-data representation is used in \cite{pola2008approximately,zamani2012symbolic}, inter-sample behaviours are not considered in such approximations. The authors of \cite{ozay2013computing} considered partition-based over-approximations of nonlinear systems for synthesizing controllers for temporal logic specifications. Because no time-discretization is used, correctness guarantee is proved for continuous-time trajectories. In \cite{liu2016finite}, the authors proposed abstractions of continuous-time nonlinear systems using grid-based approximations. A salient feature of such abstractions is that they under-approximate the control space so that all controls used by the abstractions can be implemented by the original system. At the same time, they over-approximate the reachable sets of the original system under a control so that correctness can be guaranteed (behaviours of the original system are included by the behaviours of the abstract system). In addition, the work in \cite{liu2016finite} also tackled the problem of synthesizing robust controllers and reasoned inter-sampling behaviours so that correctness is proved in continuous-time semantics of linear temporal logic.  
In \cite{reissig2017feedback}, the authors proposed feedback refinement relations that can be used for control design for systems modelled by difference inclusions. This system relation has the same feature of under-approximating the control space, while over-approximating the reachable sets of the original system. 
 Nonetheless, all the above mentioned abstractions are sound but not complete, with the exception of \cite{tabuada2006linear,pola2008approximately}, where additional assumptions on system dynamics are needed (controllable linear and incrementally stable, respectively). In \cite{liu2017robust}, a notion of completeness for abstractions of discrete-time nonlinear systems is proved by way of robustness (termed as robust completeness). It is shown that with sufficient computational sources, one can construct a finite transition system that robustly abstracts a discrete-time nonlinear system and, at the same time, is robustly abstracted by a slightly perturbed version of the same system. The case for continuous-time control system, however, is left open. In this paper, we prove that robustly complete abstractions of sampled-data continuous-time control systems also exist under a mild assumption (i.e. local Lipschitz continuity) on system dynamics and use this to show decidability of robust realization of temporal logic formulas for continuous-time nonlinear systems by using a sample-and-hold control strategy. We also note that in \cite{li2018invariance,li2018robustly} robust completeness is achieved for invariance and reachability specifications using interval analysis for direct control synthesis on the continuous state space without first constructing abstractions.

The result in the paper is of potential interest for connecting validated computation in numerical analysis with formal methods for control design. Numerical analysis plays a paramount role in all branches of science and engineering. Validated computation \cite{tucker2011validated,nedialkov1999validated,moore1979methods} is a branch of numerical analysis that seeks to compute with guarantees. It  
seems natural to ask to what extent validated computation can help with control systems design with formal guarantees. There is a fundamental difference, however, between the convergence analysis of numerical methods (or reachability analysis) for differential equations and the type of completeness results one would like to seek for control synthesis. The former is often done on a finite time horizon and for convergence to a fixed trajectory (or a set of trajectories). The latter is on the closeness of system behaviours under a control strategy over an infite time horizon. In essence, the difference lies between analysis and design: Numerical analysis and validated computation are geared towards analysis, whereas formal methods are often used for designing controllers in this setting. The system relations mentioned above can conveniently bridge this difference by constructing arbitrarily close system approximations with the help of validated computational tools, including those for accurate reachability analysis of dynamical systems \cite{girard2005reachability,rungger2018accurate,althoff2008reachability,dang2010accurate,kong2015dreach}.

\section{Problem formulation}

\subsection{Continuous-time control system}

Consider a \emph{continuous-time nonlinear control system} of the form:
\begin{equation}\label{eq:sys}
x'= f(x,u),
\end{equation}
where $x\in X\subset \mathbb{R}^n$ denotes the system state and $u\in U\subset \mathbb{R}^m$ denotes the control input. We assume that $f:\,\Real^n\times \Real^m\ra\Real^n$ satisfies the basic regularity assumptions (e.g. local Lipschitz continuity) such that, given any sufficiently regular control input signal and any initial condition, there exists a unique local solution to (\ref{eq:sys}). 

A \emph{trajectory} of (\ref{eq:sys}) is a pair $(\x,\u)$, where $\x:\,\Real^+\ra X$ is a state trajectory, $\u:\,\Real^+\ra U$ is an input trajectory, and $(\x,\u)$ satisfies (\ref{eq:sys}) in the sense that $\x'(t)=f(\x(t),\u(t))$ for all $t\ge 0$.  

A \emph{(sample-and-hold) control strategy} with sampling period $\tau>0$ for (\ref{eq:sys}) is a partial function of the form:
\begin{equation}\label{eq:control}
\sigma(x_0,\cdots,x_i) = u_i\in U,\; \forall i=0,1,2,\cdots,
\end{equation}
where $x_0,\cdots,x_i$ is a finite sequence of sampled states taken at sampling times $t_0=0,\cdots,t_i$ and $u_i$ is a constant control input. The sampling times $t_0$, $t_1$, $t_2$, $\cdots$ satisfy $t_{i+1}-t_i=\tau$ for all $i\ge 0$, where $\tau>0$ is the sampling period that represents the duration for which the constant $u_i$ is applied to the system. 

A \emph{$\sigma$-controlled trajectory} is a trajectory $(\x,\u)$ resulting from executing the control strategy $\sigma$, where $\u$ is defined by 
$
\u(t) = u_i$ for $t\in [t_i,t_{i+1})$,
where $t_i = i\tau$ and $u_i$ is determined by (\ref{eq:control}).

Given a positive integer $N$, a control strategy $\sigma$ is said to have \emph{dwell time $N$}, if each control input $u_i$ is used for a multiple of $N$ times, that is, if $i=m N$ for some integer $m$, then 
\begin{equation}\label{eq:dwell}
u_i = u_{i+1} = \cdots = u_{i+N-1}. 
\end{equation}
This can be easily encoded by a control strategy with a simple counter. This seemingly peculiar definition plays a role later on in proving completeness for any fixed, but not necessarily small, sampling period. 

\subsection{$\delta$-perturbed control system}

Given a scalar $\delta\ge 0$, a \emph{$\delta$-perturbation} of the {continuous-time nonlinear control system} \eqref{eq:sys} is the differential inclusion 
\begin{equation}\label{eq:sys2}
x'(t) \in f(x,u) + \delta B. 
\end{equation}
A \emph{trajectory} of (\ref{eq:sys2}) is a pair $(\x,\u)$, where $\x:\,\Real^+\ra X$ is a state trajectory, $\u:\,\Real^+\ra U$ is an input trajectory, and $(\x,\u)$ satisfies (\ref{eq:sys2}) in the sense that $\x'(t)\in f(\x(t),\u(t)) + \delta B$ for all $t\ge 0$. 

We call system (\ref{eq:sys}) the nominal system and denote it by $\S$. The $\delta$-perturbation of $\S$ defined by \eqref{eq:sys2} is denoted by $\S_\delta$. Apparently, $\S_0$ is exactly $\S$.

\subsection{Specifications and labelling function}

We use linear temporal logic (LTL) to specify system properties. We omit the syntax and semantics LTL formulas for limited space. For these technical details, readers are referred to \cite{baier2008principles} or the Appendix of this paper. In this section, we  emphasize the role of the lablelling function in connecting a concrete state space to an abstract logic formula. The semantics of LTL over continuous-time and discrete-time signals  are achieved by  a \emph{labelling function} $L:\,\Real^n\ra 2^{\Pi}$ that maps a state to a set of propositions (i.e., a subset of $\Pi$) that hold true for this state.

\subsubsection{Strengthening of labelling function} 

In the following, we need to reason about satisfaction of LTL formulas by continuous-time trajectories and by discrete-time sequences, and in particular, the implication between the two. For this purpose, we need to introduce the notion of an $\eps$-strengthening of a labelling function \cite{liu2017robust}. For $\eps>0$, a labelling function $\hat{L}:\,\Real^n\ra 2^\Pi$ is said to be the \emph{$\eps$-strengthening} of another labelling function $L:\,\Real^n\ra 2^\Pi$, if $\pi\in \hat{L}(x)$ \emph{if and only if} $\pi\in L(y)$ for all $y\in x+\eps\mathbb{B}$. With a possible abuse of notion, we sometimes use $L_\eps$ to denote the $\eps$-strengthening of $L$.

The following proposition relates different strengthening of labelling functions, which is used later in the proof of the main theorem. 

\begin{prop}\label{prop:strengthening}
Suppose that $\eps_2\ge \eps_1$. Let $L_{\eps_1}$ be the $\eps_1$-strengthening of a labelling function $L:\,\Real^n\ra 2^{\Pi}$. Let $(L_{\eps_1})_{\eps_2}$ be the $\eps_2$-strengthening of $L_{\eps_1}$. Let $L_{\eps_1+\eps_2}$ be the $(\eps_1+\eps_2)$-strengthening of $L$. Then $L_{\eps_1+\eps_2}(x)\subset (L_{\eps_1})_{\eps_2}(x)$ for all $x\in\Real^n$. 
\end{prop}

\begin{proof}
Pick $\pi\in L_{\eps_1+\eps_2}(x)$, then $\pi\in L(y)$ for all $y\in x + (\eps_1+\eps_2)\mathbb{B}$. To prove $\pi \in (L_{\eps_1})_{\eps_2}(x)$, we have to show that $\pi \in L_{\eps_1}(z)$ for all $z\in x + \eps_2\B$. Fix any such $z$, we verify $\pi \in L_{\eps_1}(z)$ by showing that $\pi\in L(w)$ for all $w\in z+\eps_1 \B$. This is true by the triangle inequality $\abs{w-x}\le \abs{w-z}+\abs{z-x} \le \eps_1+\eps_2$. 
\end{proof}

\subsection{Robust decidability of sampled-data control}

Given a temporal logic formula $\phi$ together with a labelling function $L$, we would like to design a sample-and-hold control strategy such that the resulting continuous-time trajectories of $\S_\delta$ satisfy $(\phi,L)$. If such a control strategy exist, we say $(\phi,L)$ is \emph{realizable} for $\S_\delta$ (by a sample-and-hold control strategy).

We formulate the robust decidability problem for control of system (\ref{eq:sys}) as follows. 

\begin{problem}[Robust decidability]\label{prob:main}
Given a temporal logic formula $\phi$ with a labelling function $L$, a sampling period $T>0$, numbers $\delta_2>\delta_1\ge 0$ and $\eps>0$, decide whether one of the following is true:
\begin{itemize}
\item There exists (and one can algorithmically construct) a sample-and-hold control strategy with sampling period $T$ for $\S_{\delta_1}$ to realize the specification $(\phi,L)$; or
\item There does not exist a sample-and-hold control strategy with sampling period $T$ for $\S_{\delta_2}$ to realize the specification $(\phi,L_\eps)$. 
\end{itemize}
\end{problem}

We shall give a positive answer to this question when $f$ satisfies the following assumption. 

\begin{assumption}
The sets $X$ and $U$ are compact and $f$ is locally Lipschitz in both $x$ and $u$.
\end{assumption}

With this assumption, it follows that there exists a constant $L\ge 0$ 
such that 
\begin{align*}
\abs{f(x,u)-f(y,u)}&\le L\abs{x-y},\\
\abs{f(x,u)-f(x,v)}&\le L\abs{u-v}, 
\end{align*}
for all $x,y\in X$ and $u,v\in U$. 

\begin{remark}
We focus on sampled-data control strategies in this paper. The restriction is not a severe one, as most of the literature on continuous-time control synthesis considers time-discretized versions of continuous-time plants \cite{pola2008approximately,pola2009symbolic,reissig2017feedback,zamani2012symbolic,liu2016finite}. Here the sampling time is not fixed \emph{a priori}, but considered as a design parameter in continuous-time control synthesis. The design of sampled-data control strategies is also favourable in practice, because such control strategies are readily implementable on digital controllers. We would also like to highlight that, despite the use of time-discretization for the control signals, the reasoning of correctness, with respect to satisfaction of temporal logic formulas, is in continuous time.
\end{remark}

\begin{remark}
In our problem formulation, the sampling period $T$ can be an arbitrarily but fixed number. The restriction to a fixed sampling period is not a severe one either. In fact the main result of this paper shows that Problem \ref{prob:main} can be answered for \textbf{each} sufficiently small sampling period. We also proved that there exists a single decision procedure to decide robust realizability of a specification for \textbf{all} sampling periods greater than a threshold value (e.g. a lower bound limited by the sampling frequency of the sensor). Nonetheless, from a technical point of view, we are not able to prove decidability in the following sense: decide one of following (1) there exists a sample-and-hold control strategy for $\S_{\delta_1}$ to realize the specification $(\phi,L)$; or (2) there does not exist a sample-and-hold control strategy for $\S_{\delta_2}$ to realize the specification $(\phi,L_\eps)$. We leave this as an open problem.  
\end{remark}

\begin{remark}
We only consider trajectories that stay in the set $X$ for all $t\ge 0$. This is technically very easy to enforce as a safety specification, $\Box X$. When constructing abstractions, any out of domain transitions need to be encoded as such so that correctness is preserved and all the trajectories produced by a synthesized controller will satisfy the specification as well staying in the set $X$ for all $t\ge 0$.
\end{remark}

\section{Transition Systems and Finite abstractions}

In this section, we define finite abstractions of $\S_\delta$ that can be used to synthesize sampled-data control strategies for $\S_\delta$. 

\subsection{Transition systems}

\begin{definition}\em \label{def:ts}
  A \emph{transition system} is a tuple
  \begin{displaymath}
    \T=(Q,A,R), 
  \end{displaymath}
where
  \begin{itemize}
  \item $Q$ is the set of states;  
  \item $A$ is the set of actions;
  \item $R \subseteq Q \times A \times Q$ is the transition relation;
  \end{itemize}
\end{definition}

For each action $a\in A$ and $q\in Q$, we define the $a$-successor of $q$ by 
$$
\text{Post}_{\T}(q,a)=\set{q':\,q'\in Q\text{ s.t. }(q,a,q')\in R}. 
$$
To simplify the presentation, we assume in this paper that, for the transition systems under consideration, every action is admissible for every state in the sense that $\text{Post}_{\T}(q,a)\neq \emptyset$ for all $q\in Q$ and all $a\in A$.

An \emph{execution} of $\T$ is an infinite alternating sequence of states and actions  
$\rho=q_0,a_0,q_1,a_1,q_2,a_2,\cdots,$
where $q_0$ is some initial state and $(q_i,a_i,q_{i+1})\in R$ for all $i\ge 0$. 
The \emph{path} resulting from the execution $\rho$ above is the sequence
$
\text{Path}(\rho)=q_0,q_1,q_2\cdots.
$
A \emph{control strategy} $\mu$ for a transition system $\T$ is a partial function $\mu:\,(q_0,q_1,\cdots,q_i)\mapsto a_i$ that maps the state history to the next action. An \emph{$\mu$-controlled execution} of a transition system $\T$ is an execution of $\T$, where for each $i\geq 0$, the action $a_i$ is chosen according to the control strategy $\mu$; $\mu$-controlled paths are defined in a similar fashion. A dwell-time control strategy is defined in the same way as that for $\S_\delta$ in (\ref{eq:dwell}). 

\subsection{Transition systems for sampled-data control systems}

With a fixed sampling period $\tau>0$, we define the transition system representation of $\S_\delta$ as follows. 

\begin{definition}\em  \label{def:rtcts}
The system $\S_\delta$ with a sampling period $\tau>0$ can be interpreted as a transition system 
$$
\T_{\delta,\tau}= (Q,A,R),  
$$
by defining 
\begin{itemize}
\item $Q=X$; 
\item $A=U$;
\item $(x_0,u,x_1)\in R$ if and only if there exists a trajectory $\x:[0,\tau]\ra X$ such that $x(0)=x_0$, $x_1=x(\tau)$, and $x'(s)\in f(x(s),u)+\delta B$ for all $s\in[0,\tau]$. 
\end{itemize}
\end{definition}

We say that an execution $\rho$ of $\T_{\delta,\tau}$ satisfies an $\ltlx$ formula $\phi$ with a labelling function $L$, written as $\rho\vDash(\phi,L)$, if and only if $\text{Path}(\rho)\vDash(\phi,L)$. For a control strategy $\mu$ for $\T_{\delta,\tau}$, if all $\mu$-controlled executions of $\T_{\delta,\tau}$ satisfy $\phi$ with respect to $L$, we write $(\T_{\delta,\tau},\mu)\vDash (\phi,L)$. If such a control strategy $\mu$ exists, we say that $(\phi,L)$ is \emph{realizable} for $\T_{\delta,\tau}$.

The following proposition relates realizability of a temporal logic formula $\phi$ on a continuous-time control system with sampled-data control strategies of different sampling periods.  
\begin{prop}\label{prop:finersample}
Let $\phi$ be a temporal logic formula over $\Pi$ and $L:\,X\ra 2^\Pi$ be a labelling function. Suppose that $T=N\tau$, where $N$ is a positive integer. 
\begin{enumerate}
\item If $(\phi,L)$ is realizable for $\S_{\delta}$ with a sampled-data control strategy with sampling period $T$, then $(\phi,L)$ is realizable for $\S_{\delta}$ with a sampled-data control strategy with sampling period $\tau$ and dwell time $N$. 
\item Conversely, if $(\phi,L)$ is realizable for $\S_{\delta}$ with a sampled-data control strategy with sampling period $\tau$ and dwell time $N$, then $(\phi,L)$ is realizable for $\S_{\delta}$ with a sampled-data control strategy with sampling period $T$.
\end{enumerate}
\end{prop}

The proof of the above proposition is straightforward. A dwell-time $N$ control strategy with sampling period $\tau$ corresponds exactly to a sampled-data control strategy with sampling period $T=N\tau$. Note that a control strategy can have memory and can easily encode consecutive use of the same control input for a finite number of times. 
 
By the assumption that $X$ and $U$ are compact sets, we can define $M=\max_{x\in X,\,u\in U}\abs{f(x,u)}$. The following proposition relates realizability of a temporal logic formula $\phi$ on a sampled-data transition system ($\T_{\delta,\tau}$) and a continuous-time system ($\S_\delta$). The main technical part is to show how discrete-time and continuous-time semantics of temporal logic formulas imply each other. 

\begin{prop}[Inter-sampple correctness]\label{prop:intersample}
Let $\phi$ be a temporal logic formula over $\Pi$. Let $L:\,X\ra 2^\Pi$ be a labelling function and $L_\eps$ be an $\eps$-strengthening of $L$. Suppose that $\eps\ge (M+\delta)\tau/2$. 
\begin{enumerate}
\item If $(\phi,L_\eps)$ is realizable for $\T_{\delta,\tau}$ with a dwell-time $N$ control strategy, then $(\phi,L)$ is realizable for $\S_\delta$ with a sampled-data control strategy with sampling period $\tau$ and dwell-time $N$. 
\item Conversely, if $(\phi,L_\eps)$ is realizable for $\S_\delta$ with a sampled-data control strategy with sampling period $\tau$  and dwell-time $N$, then $(\phi,L)$ is realizable for $\T_{\delta,\tau}$ with a dwell-time $N$ control strategy. 
\end{enumerate}
\end{prop}

\begin{proof}
The proof can be found in the Appendix. 
\end{proof}

\subsection{Abstraction}

We define control abstraction of transition system that preserves realizability of temporal logic specifications. 

\begin{definition}\label{def:ra}\em 
Given two transition systems 
$$\T_1=(Q_1,A_1,R_1),$$
and 
$$\T_2=(Q_2,A_2,R_2),$$
a relation $\alpha\subset Q_1\times Q_2$ is said to be an \emph{abstraction} from $\T_1$ to $\T_2$, if the following conditions are satisfied: 
\begin{itemize}
\item[(i)]  for all $q_1\in Q_1$, there exists $q_2\in Q_2$ such that $(q_1,q_2)\in \alpha$ (i.e., $\alpha(q_1)\neq \emptyset$); 
\item[(ii)] for all $q_2\in Q_2$ and $a_2\in A_2$, there exists $a_1\in A_1$ 
such that 
\begin{equation}\label{eq:over}
\alpha(\text{Post}_{\T_1}(q_1,a_1))\subset \text{Post}_{\T_2}(q_2,a_2);
\end{equation} 
for all $q_1\in \alpha^{-1}(q_2)$. 
\end{itemize}
If such a relation $\alpha$ exists, we say that $\T_2$ \emph{abstracts} $\T_1$ and write $\T_1 \preceq_{\alpha} \T_2$ or simply $\T_1\preceq \T_2$. When both $Q_1$ and $Q_2$ are subsets of $\Real^n$, we say that $\alpha$ is of granularity $\eta>0$, if for every $q_2\in Q_2$, $\alpha^{-1}(q_2)\subset q_2 + \eta \B$. 
\end{definition}

The following proposition shows that the abstraction relation defined above is sound in the sense of preserving realization of temporal logic specifications. 

\begin{prop}[Soundness]\label{prop:soundness}
Consider transition systems $\T_1=(Q_1,A_1,R_1)$ and $\T_2=(Q_2,A_2,R_2)$ such that $\T_1 \preceq_{\alpha} \T_2$. Suppose that $Q_1$ and $Q_2$ are subsets of $X$. Let $L:\,X\ra 2^\Pi$ be a labelling function. Let $N$ be a positive integer. 
\begin{itemize}
\item Suppose that $\alpha$ is proposition preserving with respect to $L$ in the sense that $L(q_2)\subset L(q_1)$ for all $(q_1,q_2)\in \alpha$. Then $(\phi,L)$ is realizable for $\T_2$ implies that $(\phi,L)$ is realizable for $\T_1$. 
\item Suppose that $\alpha$ is of granularity $\eta>0$. Let $L_\eta$ be an $\eta$-strengthening of $L$. Then $(\phi,L_\eta)$ is realizable for $\T_2$ implies that $(\phi,L)$ is realizable for $\T_1$. 
\end{itemize}
Moreover, a dwell-time $N$ strategy for $\T_2$ can be implemented by a dwell-time $N$ strategy for $\T_1$. 
\end{prop}

\begin{proof}
The proof is similar to the proof of Theorem 1 in \cite{liu2017robust}. Additional consideration has to be given to the dwell-time requirement and the separate cases of proposition preserving and finite-granularity abstractions. 
\end{proof}

A strengthening of labelling function is needed here if the abstraction is not proposition preserving with the original labelling function. For instance, if a proposition is defined as a semialgebraic set
 of the form $\set{x\in X:\,g(x)\le 0}$, where $g$ is a polynomial function.   
 There is no guarantee that a grid-based partition will preserve this proposition. The second part of the proposition can be used, where a strengthening of labelling function is needed to account for this mismatch. 

\section{Robustly Complete Abstraction and Robust decidability}

In this section, we prove that sampled-data control for nonlinear system is robustly decidable. 

\subsection{Robustly complete abstraction}

The key technical result for proving robustly decidability of sampled-data control for nonlinear system is the following result on the possibility of constructing an arbitrarily accurate abstraction of the nonlinear system in the sense that for any $\delta_2>\delta_1\ge 0$, one can find a finite transition system $\T$ such that $\T$ abstracts $\S_{\delta_1}$ while $\S_{\delta_2}$ abstracts $\T$. Hence, realizability of a specification by $\S_{\delta_2}$ would imply realizability of the same specification by $\S_{\delta_1}$.

\begin{theorem}[Robust completeness]\label{thm:complete}
Given any $\delta_2>\delta_1\ge 0$, we can choose $\tau>0$ and compute a finite transition system $\T$ such that 
$$
\T_{\delta_1,\tau}  \preceq \T \preceq \T_{\delta_2,\tau}.
$$
\end{theorem}

\begin{proof}
We construct $\T=(Q,A,R)$ as follows. Let $\eta>0$ and $\mu>0$ be parameters to be chosen. 
Let $Q$ consist of the centres of the grid cells in $[\Real^n]_{\eta}$ that have a non-empty intersection with $X$. Let $A$ consist of the centres of the grid cells in $[\Real^m]_{\mu}$ that have a non-empty intersection with $U$. Because $U$ and $X$ are compact sets, $Q$ and $A$ are both finite.
We define a relation $\alpha\subset X\times Q$ by $(x,q)\in \alpha$ if and only if $\abs{x-q}\le \frac{\eta}{2}$. Clearly, $\alpha^{-1}$ is a relation on $Q\times X$. Define $R\subset (Q,A,Q)$ by $(q,a,q')\in R$ if and only if  
\begin{align}
&\abs{q' - (q+\tau f(q,a))} \notag\\
&\qquad\le \frac{\eta}{2} + \frac{\eta}{2}e^{L\tau} + (\frac{\delta_1}{L}  + \frac{\mu}{2})(e^{L\tau}-1) \notag\\
&\qquad\qquad+\frac{M(e^{L\tau}-L\tau-1)}{L}.  \label{eq:relation}
\end{align}
We show that, if $\eta$, $\mu$, and $\tau$ are chosen sufficiently small, we have 
$$
\T_{\delta_1,\tau}  \preceq_{\alpha} \T \preceq_{\alpha^{-1}} \T_{\delta_2,\tau}.
$$
Condition (i) in Definition \ref{def:ra} is clearly satisfied by both $\alpha$ and $\alpha^{-1}$. 

We verify that condition (ii) holds for $\T_{\delta_1,\tau}  \preceq_{\alpha} \T$, that is, for $q\in Q$ and $a\in A$,  there exists $u\in U$ such that 
\begin{equation}\label{eq:post1}
\alpha(\text{Post}_{\T_{\delta_1,\tau}}(x,u))\subset \text{Post}_{\T}(q,a);
\end{equation}
for all $x\in \alpha^{-1}(q)$. Pick $u\in U$ with $\abs{u-a}\le\frac{\mu}{2}$. Given $x'\in \text{Post}_{\T_{\delta_1,\tau}}(x,u)$, there exists a trajectory $\x:[0,\tau]\ra X$ such that $\x(0)=x$, $\x(\tau)=x'$, and $\x'(s)\in f(\x(s),u)+\delta_1 B$ for all $s\in [0,\tau]$. Define $\x_\tau(t)=q+t f(q,u)$ for $t\in [0,\tau]$. We have
\begin{align}
&\abs{\x'(t)-\x'_\tau(t)} \notag\\
&\le \abs{f(\x(t),u)-f(q,a)} + \delta_1 \notag\\
&\le \abs{f(\x(t),u)-f(\x_\tau(t),u)} + \abs{f(\x_\tau(t),u)-f(q,u)}\notag\\
&\qquad\qquad+\abs{f(q,u)-f(q,a)}+ \delta_1\notag\\
&\le L\abs{\x(t)-\x_\tau(t)} + L\abs{\x_\tau(t)-q} +L\abs{u-a}+ \delta_1\notag\\
& \le L\abs{\x(t)-\x_\tau(t)} + LMt + \frac{L\mu}{2} + \delta_1,\quad t\in [0,\tau].\label{eq:deviate}
\end{align}
By Gronwall's inequality (see, e.g., \cite{aubin2012differential}), we have
\begin{align*}
&\abs{x'-(q+\tau f(q,u))}=\abs{\x(\tau)-\x_\tau(\tau)} \\
&\le \abs{x-q}e^{L\tau} + \int_0^\tau ( LM s + \frac{L\mu}{2} + \delta_1) e^{L(\tau-s)}ds\\
&\le \frac{\eta}{2}e^{L\tau} + (\frac{\delta_1}{L} + \frac{\mu}{2})(e^{L\tau}-1) +\frac{M(e^{L\tau}-L\tau-1)}{L}.
\end{align*}
By (\ref{eq:relation}), this shows $\alpha(x')\subset \text{Post}_{\T}(q,a)$. Hence (\ref{eq:post1}) holds. 

We next verify that condition (ii) holds for $\T \preceq_{\alpha^{-1}} \T_{\delta_2,\tau}$, that is, for $x\in X$ and $u\in U$, there exists $a\in A$ such that 
\begin{equation}\label{eq:post2}
\alpha^{-1}(\text{Post}_{\T}(q,a))\subset \text{Post}_{\T_{\delta_2,\tau}}(x,u);
\end{equation}
for all $q\in \alpha(x)$. Pick $a$ be the center of the grid cell in $[\Real^m]_{\mu}$ that contains $u$. Given $y'\in \alpha^{-1}(\text{Post}_{\T}(q,a))$, there exists $q'\in\text{Post}_{\T}(q,a)$ such that $\abs{y'-q'}\le\frac{\eta}{2}$. By the definition of $\text{Post}_{\T}(q,a)$, we have 
\begin{align*}
\abs{q'-(q+\tau f(q,a))}&\le  \frac{\eta}{2} + \frac{\eta}{2}e^{L\tau} + (\frac{\delta_1}{L} + \frac{\mu}{2})(e^{L\tau}-1)\\
& +\frac{M(e^{L\tau}-L\tau-1)}{L}.
\end{align*}
Consider the trajectory $\x:[0,\tau]\ra X$ such that $\x(0)=x$, $\x(\tau)=x'$, and $\x'(s)\in f(\x(s),u)$. By a similar argument as in (\ref{eq:deviate}), we can show 
\begin{align*}
\abs{x'-(q+\tau f(q,a))}&\le \frac{\eta}{2}e^{L\tau} + \frac{\mu}{2}(e^{L\tau}-1)\\
&\qquad +\frac{M(e^{L\tau}-L\tau-1)}{L}.
\end{align*}
Hence, by the triangle inequality, 
\begin{equation}\label{eq:tri1}
\abs{y'-x'}\le \eta+ \eta e^{L\tau} + (\frac{\delta_1}{L}+\mu)(e^{L\tau}-1) +\frac{2M(e^{L\tau}-L\tau-1)}{L}
\end{equation}
Define
$$
\z(\theta) = \x(\theta) + \frac{\theta}{\tau} [y'-x'],\quad \theta\in [0,\tau].
$$
Then $\z(0)=\x(0)=x$ and $\z(\tau)=y'$, and
\begin{align}\label{eq:z}
\z'(\theta) \in f(\x(\theta),u) + \frac{1}{\tau}[y'-x'].
\end{align}
Note that 
\begin{align}\label{eq:tri2}
|\z(\theta) - \x(\theta)| & = \big\vert \frac{\theta}{\tau} [y'-x'] \big\vert \le \abs{y'-x'},\quad \theta \in [0,\tau]. 
\end{align}
Since $0\le\delta_1<\delta_2$, we can choose $\tau$, $\mu$, $\eta$ sufficiently small such that
{\small
\begin{equation}\label{eq:margin}
[\eta+ \eta e^{L\tau} + (\frac{\delta_1}{L}+\mu)(e^{L\tau}-1) +\frac{2M(e^{L\tau}-L\tau-1)}{L}][L+\frac{1}{\tau}]<\delta_2. 
\end{equation}
}
To see this is possible, choose, e.g. $\eta=\tau^2$ and $\mu=\tau$, and note that the limit of the left-hand side as $\tau\ra 0$ is given by
$$
\lim_{\tau \ra 0}\delta_1 \frac{e^{L\tau-1}}{L\tau} = \delta_1. 
$$
It follows from (\ref{eq:tri1})--(\ref{eq:margin}) and Lipschitz continuity of $f$ that
$$
\z'(\theta) \in f(\z(\theta),u) + \delta_2 B. 
$$
Hence $y'\in \text{Post}_{\T_{\delta_2,\tau}}(x,u)$ and (\ref{eq:post2}) holds. 
\end{proof}

\begin{remark}
In the proof, we choose the simplest possible validated bounds on a one-step reachable set, i.e. a forward Euler scheme with an error bound. This suffices to prove the required convergence to show approximate completeness. With the template provided by the proof of Theorem \ref{thm:complete}, one can in fact use any accurate over-approximation of the one-step reachable set for $\S_{\delta_1}$ to replace (\ref{eq:relation}) for defining the transitions in $\T$ and then show that this over-approximation is contained in the actual one-step reachable set of $\S_{\delta_2}$. 
\end{remark}

\begin{remark}
Theorem \ref{thm:complete} (as well as the problem formulation in the paper) only considers sample-and-hold control strategies. To prove a similar result for a more general set of signals $\mathcal{U}$, one would need to prove that, for each granularity $\mu>0$, there exists a finite subset of signals $\mathcal{A}$ that can approximate the set of signals $\mathcal{U}$ to a precision $\mu$ in the sense that, for every $\u\in \U$, there exists $\mathbf{a}\in\mathcal{A}$ such that $\norm{u-\mathbf{a}}\le \mu$, where $\norm{\cdot}$ denote the maximum norm. 
\end{remark}

\subsection{Robust decidability}

The following theorem is an immediate consequence of Theorem \ref{thm:complete} and states that sampled-data control for nonlinear system is robustly decidable. 

\begin{theorem}[Robust decidability]\label{thm:decidability}
Given a temporal logic specification $\phi$, a sampling period $T>0$, any $\delta_2>\delta_1\ge 0$, and any $\eps>0$. Let $L:\, X\ra 2^{\Pi}$ be a labelling function and $L_\eps$ be an $\eps$-strengthening of $L$. Then there exists a decision procedure that determines one of the following: 
\begin{itemize}
\item there exists (and one can algorithmically construct) a sample-and-hold control strategy with sampling period $T$ such that $(\phi,L)$ is realizable for $\S_{\delta_1}$; or 
\item $(\phi,L_\eps)$ is not realizable for $\S_{\delta_2}$ with a sample-and-hold control strategy with sampling period $T$.  
\end{itemize}
\end{theorem}

\begin{proof}
Suppose that $(\phi,L_\eps)$ is realizable for $\S_{\delta_2}$ with a sampled-data control strategy with sampling period $T$.   
Let $N$ be a positive integer and $\tau= \frac{T}{N}$. Let $\eps_1=\frac{(M+\delta_1)\tau}{2}$ and $\eps_2=\frac{(M+\delta_2)\tau}{2}$. Choose $\tau$ sufficiently small such that 
\begin{equation}\label{eq:margin2}
\frac{(2M+\delta_1+\delta_2)\tau}{2} = \eps_1+\eps_2 \le \eps.
\end{equation}
Let $L_{\eps_1}$ be the $\eps_1$-strengthening of $L$. Let $(L_{\eps_1})_{\eps_2}$ be the $\eps_2$-strengthening of $L_{\eps_1}$. Let $L_{\eps_1+\eps_2}$ be the $(\eps_1+\eps_2)$-strengthening of $L$. By the definition of strengthening a labeling function and Proposition \ref{prop:strengthening}, we have $L_\eps(x)\subset L_{\eps_1+\eps_2}(x) \subset (L_{\eps_1})_{\eps_2}(x)$ for all $x\in X$. Hence, by the semantics of $\ltlx$, $(\phi,(L_{\eps_1})_{\eps_2}(x))$ is realizable for $\S_{\delta_2}$ with a sampled-data control strategy with sampling period $T$.

By Proposition \ref{prop:finersample}, $(\phi,(L_{\eps_1})_{\eps_2})$ is realizable for $\S_{\delta_2}$ with a sampled-data control strategy with sampling period $\tau$ and dwell-time $N$. By Proposition \ref{prop:intersample}, $(\phi,L_{\eps_1})$ is realizable for $\T_{\delta_2,\tau}$ with a dwell-time $N$ control strategy, because $\eps_2\ge\frac{(M+\delta_2)\tau}{2}$ (indeed equal).   
Construct $\T$ by Theorem \ref{thm:complete} so that  
$$
\T_{\delta_1,\tau}  \preceq \T \preceq \T_{\delta_2,\tau}.
$$
By Theorem \ref{prop:soundness}, $(\phi,L_{\eps_1})$ is realizable for $\T$ and hence also for $\T_{\delta_1,\tau}$ with a dwell-time $N$ control strategy.  By Proposition \ref{prop:intersample}, $(\phi,L)$ is realizable for $\S_{\delta_1}$ with a sampled-data control strategy with sampling period $\tau$ and dwell-time $N$, because $\eps_1\ge\frac{(M+\delta_1)\tau}{2}$. Finally, by Proposition \ref{prop:finersample} again, $(\phi,L)$ is realizable for $\S_{\delta_1}$ with a sampled-data control strategy with sampling period $T$. One can algorithmically construct such a control strategy by synthesizing a dwell-time $N$ controller  strategy for the finite transition system $\T$.
For the case there is not necessarily a proposition preserving partition, we can choose $\eps_1=\frac{(M+\delta_1)\tau+\eta}{2}$ and $\eps_2=\frac{(M+\delta_2)\tau+\eta}{2}$ to account for mismatch by an abstraction with granularity $\eta$. In this case, we can choose $\tau$ and $\eta$ sufficiently small such that 
\begin{equation}\label{eq:margin3}
\eta + \frac{(2M+\delta_1+\delta_2)\tau}{2} = \eps_1+\eps_2 \le \eps.
\end{equation}
\end{proof}

A decision diagram summarizing the argument in the proof of Theorem \ref{thm:decidability} is shown in Figure \ref{fig:decision}.

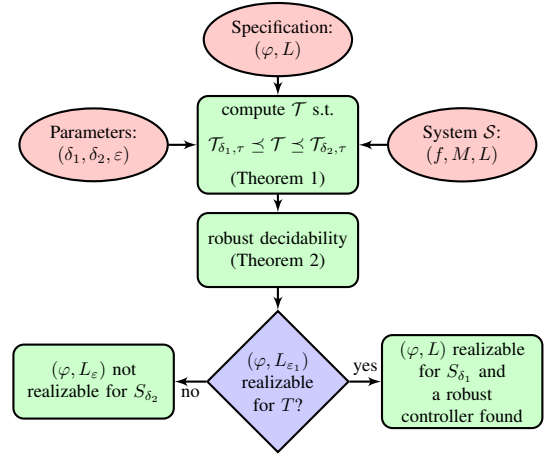
\begin{figure}
\tikzstyle{decision} = [diamond, draw, fill=blue!20, 
    text width=4.5em, text badly centered, node distance=3cm, inner sep=0pt]
\tikzstyle{block} = [rectangle, draw, fill=blue!20, 
    text width=5em, text centered, rounded corners, minimum height=4em]
\tikzstyle{greenblock} = [rectangle, draw, fill=green!20, 
    text width=5em, text centered, rounded corners, minimum height=4em]
\tikzstyle{longblock} = [rectangle, draw, fill=green!20, 
    text width=8em, text centered, rounded corners, minimum height=4em]
\tikzstyle{line} = [draw, -latex']
\tikzstyle{cloud} = [draw, ellipse,fill=red!20, text width=5em, text centered, node distance=3cm,
    minimum height=2em]

\centering    

\vskip 1cm
\begin{tikzpicture}[node distance = 2cm, auto,thick,scale=0.7, every node/.style={scale=0.7}]
    \node [longblock] (abs) {compute $\T$ s.t. 
    $$\T_{\delta_1,\tau}  \preceq \T \preceq \T_{\delta_2,\tau}$$
    (Theorem \ref{thm:complete})};
    \node [cloud, right of=abs, node distance=3.5cm] (system) {System $\S$:\\ $(f,M,L)$};
    \node [cloud, left of=abs, node distance=3.5cm] (robust) {Parameters: \\$(\delta_1,\delta_2,\eps)$};
    \node [longblock, below of=abs] (synthesis) {robust decidability (Theorem \ref{thm:decidability})};       
    \node [cloud, above of=abs, node distance=2cm] (spec) {Specification:\\ $(\phi,L)$};
    \node [decision, below of=synthesis, node distance=2.5cm] (evaluate) {$(\phi,L_{\eps_1})$ realizable for $T$?};
    \node [longblock, left of=evaluate, node distance=3.5cm] (update) { $(\phi,L_\eps)$ not realizable for $S_{\delta_2}$ };  
    \node [longblock, right of=evaluate, node distance=3.5cm] (controller) {$(\phi,L)$ realizable for $S_{\delta_1}$ and a robust controller found};
    \path [line] (abs) -- (synthesis);
    \path [line] (synthesis) -- (evaluate);
    \path [line] (evaluate) -- node {yes}(controller);
    \path [line] (evaluate) -- node {no}(update);
    \path [line] (spec) -- (abs);
    \path [line] (system) -- (abs);
    \path [line] (robust) -- (abs);
\end{tikzpicture}
\caption{A decision diagram for checking robust realizability: given a system $S$, a temporal logic specification $\phi$ with a labelling function $L$, and parameters $\delta_2>\delta_1\ge 0$ and $\eps>0$, we can decide either $(\phi,L)$ is realizable for $\S_{\delta_1}$ with a robust controller, or $(\phi,L_\eps)$ is not realizable for $\S_{\delta_2}$. This is done by checking realizability of $(\phi,L_{\eps_1})$ on $\T$, where $\T$ is a sufficiently precise abstraction of $\S_{\delta_1}$ constructed as in the proof of Theorem \ref{thm:complete} by choosing the discretization parameters $\eta,\mu,\tau$ sufficiently small according to (\ref{eq:margin}), and $\eps_1$ is chosen according to (\ref{eq:margin2}) in the proof of Theorem \ref{thm:decidability}.} \label{fig:decision}
\end{figure}

When there is no \emph{a priori} fixed sampling period for the decision process, we can formulate the robust decidability theorem as follows, where it is proved that the problem can be solved for all sufficiently small sampling periods. The proof follows exactly from the proof of Theorem \ref{thm:decidability} with $N=1$. 

\begin{theorem}[Robust decidability II]\label{thm:decidability2}
Given a temporal logic specification $\phi$, any $\delta_2>\delta_1\ge 0$, and any $\eps>0$. Let $L:\, X\ra 2^{\Pi}$ be a labelling function and $L_\eps$ be an $\eps$-strengthening of $L$. Then there exists some $\tau^*>0$ (and one can explicitly compute it) such that, for each $\tau\in (0,\tau^*]$, there exists a decision procedure that determines one of the following: 
\begin{itemize}
\item there exists (and one can algorithmically construct) a sample-and-hold control strategy with sampling period $\tau$ such that $(\phi,L)$ is realizable for $\S_{\delta_1}$; or 
\item $(\phi,L_\eps)$ is not realizable for $\S_{\delta_2}$ with a sample-and-hold control strategy with sampling period $\tau$.  
\end{itemize}
\end{theorem}

Another version of robust decidability can be formulated as follows, which says that with one procedure, one can decide robust realizability by a sample-and-hold control strategy with \emph{any} sampling period greater than a threshold value (e.g. a lower bound limited by the physical sampling frequency). 

\begin{theorem}[Robust decidability III]\label{thm:decidability3}
Given a temporal logic specification $\phi$, any $\delta_2>\delta_1\ge 0$, $\eps>0$, and $\tau^*>0$. Let $L:\, X\ra 2^{\Pi}$ be a labelling function and $L_\eps$ be an $\eps$-strengthening of $L$. Then there exists some $\tau>0$ (and one can explicitly compute it) and a decision procedure that determines one of the following: 
\begin{itemize}
\item there exists (and one can algorithmically construct) a sample-and-hold control strategy with sampling period $\tau$ such that $(\phi,L)$ is realizable for $\S_{\delta_1}$; or 
\item $(\phi,L_\eps)$ is not realizable for $\S_{\delta_2}$ with a sample-and-hold control strategy with a sampling period $T\ge \tau^*$.  
\end{itemize}
\end{theorem}

To prove Theorem \ref{thm:decidability3}, we need the following lemma, which shows that, if $\delta_2>\delta_1$, then system $\T_{\delta_1,\tau}$ can be abstracted by $\T_{\delta_2,\tau'}$ despite a slight mismatch between the sampling periods $\tau$ and $\tau'$.

\begin{lemma}\label{lem:tau}
Given any $\tau^*>0$ and $\delta_2>\delta_1\ge 0$, there exists $r^*>0$ such that
$$
\T_{\delta_1,T}\preceq_{\text{id}_{X}} \T_{\delta_2,T+r}.
$$
for all $T\ge \tau^*$ and all $\abs{r}\le r^*$, where $\text{id}_{X}\subset X\times X$ is the identity relation.  
\end{lemma}

\begin{proof}
Choose any $x\in X$ and $u\in U$. Let $x_1\in \text{Post}_{\T_{\delta_1,T}}(x,u)$. We show that $x_1\in \text{Post}_{\T_{\delta_2,T+r}}(x,u)$. By definition, there exists a trajectory $\x$ such that $\x(0)=x$, $\x(T)=x_1$, and $\x'(s)\in f(\x(s),u) +\delta_1\B$ for all $s\in [0,T]$. Let $\z(s)=\x(\frac{T}{T+r}s)$ for $s\in [0,T]$. Then $\z(0)=x$, $\z(T+r)=x_1$ and 
\begin{align*}
\z'(s) &= \frac{T}{T+r}\x'(\frac{T}{T+r}s) \in \frac{T}{T+r}f(\z(s),u) + \frac{T}{T+r}\delta_1\B \\
& \subset f(\z(s),u) -  \frac{r}{T+r}f(\z(s),u) +  \delta_1\B \\
& \subset f(\z(s),u) + (\frac{\abs{r}}{\tau^*-\abs{r}}+\delta_1)\B, 
\end{align*}
where we assumed $\abs{r}$ is sufficiently small so that $\abs{r}\le \tau^*$. Clearly, since $\delta_1<\delta_2$, we can choose $r^*>0$ so that  $\frac{\abs{r}}{\tau^*-\abs{r}}+\delta_1<\delta_2$ for all $\abs{r}\le r^*$. Hence, $\z'(s)\in f(\z(s),u) + \delta_2\B$ and $x_1=\z(T+r)\in \text{Post}_{\T_{\delta_2,T+r}}(x,u)$. 
\end{proof}
 
 Now we can present the proof of Theorem \ref{thm:decidability3}. 

\begin{proof}[Proof of Theorem \ref{thm:decidability3}]

Let $\eps_1$ and $\eps_2$ be as defined in the proof for Theorem \ref{thm:decidability}. Choose $\delta_3$ such that $\delta_2>\delta_3>\delta_1$. Let $\tau^*$, $\eta^*$, and $\mu^*$ be chosen so that (\ref{eq:margin}) and (\ref{eq:margin2}) (or (\ref{eq:margin}) and (\ref{eq:margin3}) if a proposition preserving partition is not used),  with $\delta_3$ replacing $\delta_2$ in (\ref{eq:margin}), hold for all $\tau\le \tau^*$, $\eta\le\eta^*$, and $\mu\le\mu^*$. 

Suppose that $(\phi,L_\eps)$ is realizable for $\S_{\delta_2}$ with a sampled-data control strategy with sampling period $T$. Without loss of generality, assume $\frac{\tau^*}{2}<T\le \tau^*$. Otherwise, one can divide $T$ by a positive integer number $N$ so that $\frac{T}{N}\in (\tau^*/2,\tau^*]$ and $(\phi,L_\eps)$ is realizable for $\S_{\delta_2}$ with a sampled-data control strategy with sampling period $T/N$ (with dwell-time $N$). 

Construct, by Theorem \ref{thm:complete}, $\T$ so that  
\begin{equation}\label{eq:soundnew}
\T_{\delta_3,T}  \preceq \T \preceq \T_{\delta_2,T}.
\end{equation}
Let $\tau\le \frac{\tau^*}{2}$ be chosen (guaranteed by Lemma \ref{lem:tau}) so that 
\begin{equation}\label{eq:tau}
\T_{\delta_1,T+r}\preceq_{\text{id}_{X}} \T_{\delta_3,T}.
\end{equation}
for all $\abs{r}\le\tau$. 

Let $L_{\eps_1}$, $(L_{\eps_1})_{\eps_2}$ and $L_{\eps_1+\eps_2}$ be as defined in the proof for Theorem \ref{thm:decidability}. 
By Proposition \ref{prop:intersample}, $(\phi,L_{\eps_1})$ is realizable for $\T_{\delta_2,T}$, because $\eps_2\ge\frac{(M+\delta_2)T}{2}$.  By Proposition \ref{prop:soundness} and (\ref{eq:soundnew}), $(\phi,L_{\eps_1})$ is realizable for  $\T$ and hence also for $\T_{\delta_3,T}$.
Let $m$ be the largest integer such that $m \tau\le T$. Then $\abs{m\tau-T}\le \tau$. By (\ref{eq:tau}), we obtain 
$$
\T_{\delta_1,m\tau} = \T_{\delta_1,T + (m\tau-T)}\preceq_{\text{id}_{X}} \T_{\delta_3,T}.
$$
By Proposition \ref{prop:soundness} again, $(\phi,L_{\eps_1})$ is realizable for $\T_{\delta_3,m\tau}$. By Proposition \ref{prop:intersample}, $(\phi,L)$ is realizable for $\S_{\delta_1}$ with a sampled-data control strategy with sampling period $m\tau$, because $\eps_1\ge\frac{(M+\delta_1)m\tau}{2}$. Finally, by Proposition \ref{prop:finersample}, $(\phi,L)$ is realizable for $\S_{\delta_1}$ with a sampled-data control strategy with sampling period $\tau$. One can algorithmically construct such a control strategy by synthesizing a controller strategy for the finite transition system $\T$ for $(\phi,L_{\eps_1})$. 
\end{proof}
 
We leave as an open problem to decide robust realizability by a sample-and-hold control strategy with \emph{any} sampling period.

\begin{problem}[Robust decidability]\label{prob:open}
Given a temporal logic formula $\phi$ with a labelling function $L$, numbers $\delta_2>\delta_1\ge 0$ and $\eps>0$, decide whether one of the following is true:
\begin{itemize}
\item There exists (and one can algorithmically construct) a sample-and-hold control strategy for $\S_{\delta_1}$ to realize the specification $(\phi,L)$; or
\item There does not exist a sample-and-hold control strategy for $\S_{\delta_2}$ to realize the specification $(\phi,L_\eps)$. 
\end{itemize}
\end{problem}

\section{Example}

\begin{ex}[Nonlinear car]
Consider a nonlinear car with bicycle dynamics \cite{astrom2010feedback,zamani2012symbolic} (details can be found in the Appendix). 
It can be verified that $L=1.2674$ and $M=1.5574$ give a valid Lipschitz constant and upper bound for the vector field, respectively, on the compact domains $X=[0,10]\times [0,10]\times [-\pi,\pi]$ and $U=[-1,1]\times [-1,1]$. To construct robustly complete abstractions, we can choose $\eta=\tau^2$ and $\mu=\eta$ (as in the proof of Theorem \ref{thm:complete} and likely not optimized) and then make $\tau$ as small as possible to satisfy the (\ref{eq:margin}) and (\ref{eq:margin2}). The following figure shows the changes in $\delta_2$ (labeled as $\delta$ and assuming $\delta_1=0$) and $\eps$ as the size of $\eta$ varies. While it is not surprising that the bounds given in the proof of Theorem \ref{thm:complete} are conservative, we can still see from Figure \ref{fig:result} that if we pick $\tau=0.2$, we can construct an abstraction that can be used to decide, for \emph{any} temporal logics specification, either the system is realizable for this specification, or the system perturbed by disturbance of size $\delta=0.1$ cannot realize this specification (of course, subject to an $\eps$-strengthening of the labelling function with $\eps=0.02$). 
\end{ex}

\begin{figure}[h]
\centering
\includegraphics[width=0.35\textwidth]{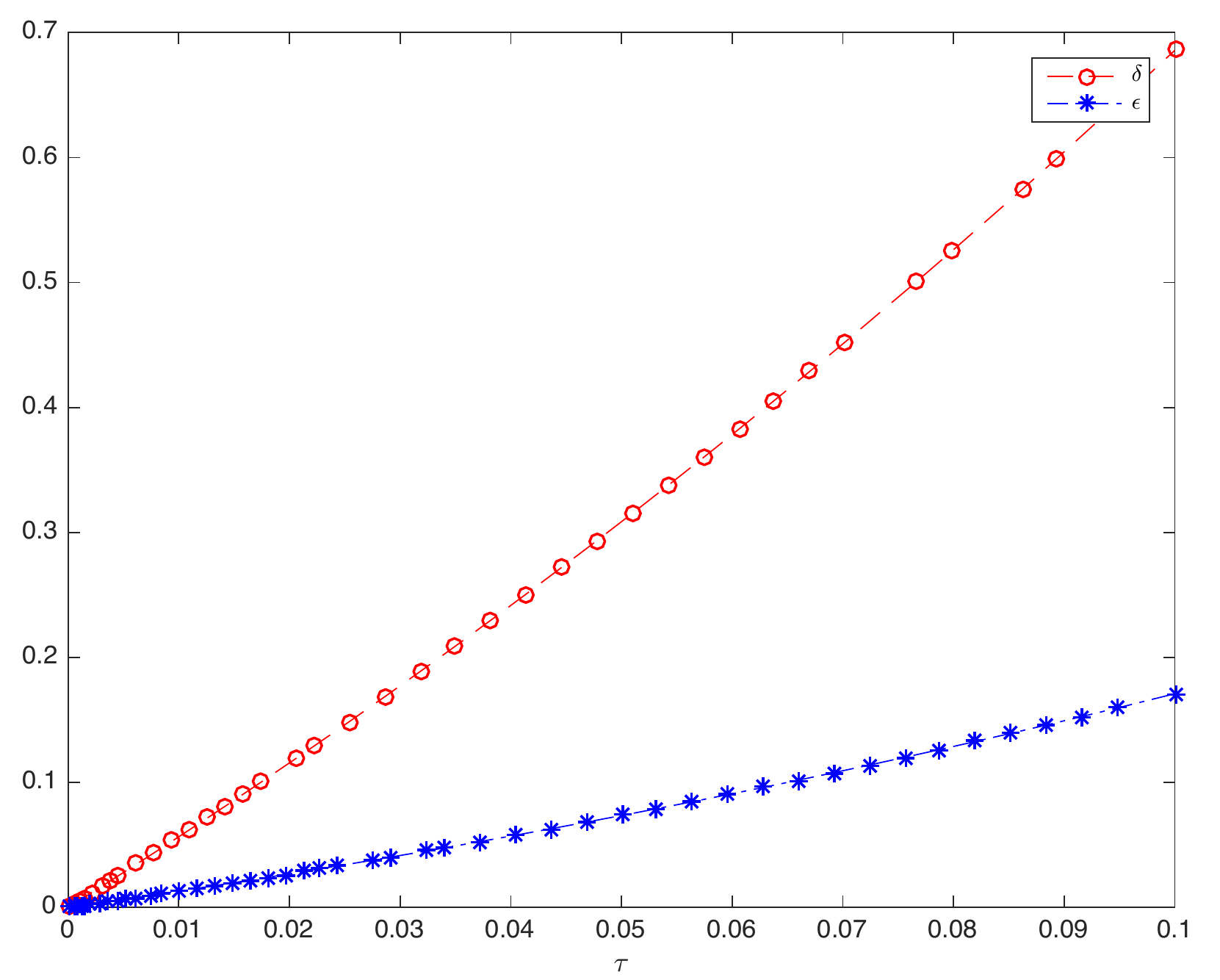}
\label{fig:result}
\caption{Size of sampling period required to achieve robustly complete abstraction using Theorem \ref{thm:complete} in view of (\ref{eq:margin}) and (\ref{eq:margin2}).}
\end{figure}
 
\section{Conclusions}

In this paper, we proved that control synthesis for sampled-data nonlinear systems with temporal logic specifications is robustly decidable in the sense that if a robust control strategy exists, then a robust control strategy can be found using a sufficiently fine discretization. The proof can be seen as an use of validated forward Euler numerical scheme. By explicitly quantifying the error bounds, we showed that it is possible to construct arbitrarily close system approximations that are suitable for control synthesis. We see the main contribution of this work as showing the existence of robustly complete abstractions for nonlinear sampled-data control systems. We also show that inter-sampling behaviours can be accounted for by having arbitrarily small strengthening of the labelling function. It is hoped that this work will motivate further research on computing tight abstractions of nonlinear control systems. In this regard, Theorem \ref{thm:complete} on robust completeness can be viewed as a potential metric on closeness of abstractions. We also leave as an open problem to decide robust realizability by a sample-and-hold control strategy with \emph{any} sampling period.

\bibliographystyle{IEEEtran}
\bibliography{test}

\newpage

\appendix

\subsection{Linear temporal logic}

We use linear temporal logic (LTL) without the next operator (denoted by $\ltlx$ \cite{baier2008principles}) to specify system properties. This logic consists of the usual propositional logic operators (e.g., \textbf{true}, \textbf{false}, {\em negation\/} ($\neg$), {\em disjunction\/} ($\vee$), {\em conjunction\/} ($\wedge$) and
{\em implication\/} ($\rightarrow$)), and additional temporal operators (e.g., {\em always\/} ($\always$), {\em eventually\/} ($\eventually$), {\em until\/} ($\mathcal{U}$) and {\em release\/} ($\mathcal{R}$)). 

\subsubsection{Syntax}

We can define the syntax of $\ltlx$ over a set of atomic propositions $\Pi$ inductively as follows:
\begin{itemize}
\item \textbf{true} and \textbf{false} are $\ltlx$ formulae; 
\item an atomic proposition $\pi\in\Pi$ is an $\ltlx$ formula; 
\item if $\phi$ and $\psi$ are $\ltlx$ formulas, then $\neg\phi$, $\phi \vee \phi$, and $\phi \U \phi$ are $\ltlx$ formulas. 
\end{itemize}

\emph{Negation Normal Form (NNF):} All $\ltlx$ formulas can be transformed into negation normal form \cite[p.~132]{clarke1999model}, where
\begin{itemize}
\item all negations appear only in front of the atomic propositions\footnote{We assume that all negations can be effectively removed by introducing new atomic propositions corresponding to the negations of current ones.}; 
\item only the logical operators \textbf{true}, \textbf{false}, $\wedge$, and $\vee$ can appear; and 
\item only the temporal operators $\U$ and $\R$ can appear, where $\R$ is defined by $\phi_1\R\phi_2\equiv \neg(\neg\phi_1\U\neg\phi_2)$, called the \emph{dual until} operator.
\end{itemize}
For syntactic convenience, we can define additional temporal operators $\Box$ and $\Diamond$ by $\Box \phi\equiv \texttt{false}\R\phi$ and $\Diamond \phi\equiv \texttt{true}\U\phi$.

\subsubsection{Semantics}

We consider two types of semantics for $\ltlx$ formulas, namely, continuous-time and discrete-time semantics. To define semantics, an atomic proposition is interpreted as a subset of the state space on which the atomic proposition holds true. This is achieved by defining a \emph{labelling function} $L:\,\Real^n\ra 2^{\Pi}$ that maps a state to a set of propositions that hold true for this state. 

{\it Continuous-time semantics of $\ltlx$:} Given a continuous-time function $\xi:\,[0,\infty)\ra \Real^n$, we define $\xi,t\vDash(\phi,L)$ with respect to an $\ltlx$ formula $\phi$ and a labelling function $L$ at time $t$ inductively as follows:
\begin{itemize}
\item $\xi,t\vDash(\pi,L)$ if and only if $\pi\in L(\xi(t))$;
\item $\xi,t\vDash(\phi_1\vee \phi_2,L)$ if and only if $\xi,t\vDash(\phi_1,L)$ or $\xi,t\vDash(\phi_2,L)$;
\item $\xi,t\vDash(\phi_1\wedge \phi_2,L)$ if and only if $\xi,t\vDash(\phi_1,L)$ and $\xi,t\vDash(\phi_2,L)$;
\item $\xi,t\vDash(\phi_1\U\phi_2,L)$ if and only if there exists $t'\ge 0$ such that $\xi,t+t'\vDash(\phi_2,L)$ and for all $t''\in [0,t')$, $\xi,t+t''\vDash(\phi_1,L)$;
\item $\xi,t\vDash(\phi_1\R\phi_2,L)$ if and only if, for all $t'\ge 0$, at least one of the following holds: $\xi,t+t'\vDash(\phi_2,L)$ or there exists $t''\in [0,t')$ such that $\xi,t+t''\vDash(\phi_1,L)$.
\end{itemize}
We write $\xi\vDash(\phi,L)$ if $\xi,0\vDash(\phi,L)$. If the labelling function is clear from the context, we simply write $\xi\vDash\phi$.

{\it Discrete-time semantics of $\ltlx$:} Given a sequence $\rho=\set{x_i}_{i={0}}^{\infty}$ in $\Real^n$, we define $\rho,i\vDash \phi$ with respect to an $\ltlx$ formula $\phi$ and a labelling function $L$ inductively as follows:
\begin{itemize}
\item $\rho,i\vDash(\pi,L)$ if and only if $\pi\in L(x_i)$;
\item $\rho,i\vDash(\phi_1\vee \phi_2,L)$ if and only if $\rho,i\vDash(\phi_1,L)$ or $\rho,i\vDash(\phi_2,L)$;
\item $\rho,i\vDash(\phi_1\wedge \phi_2,L)$ if and only if $\rho,i\vDash(\phi_1,L)$ and $\rho,i\vDash(\phi_2,L)$;
\item $\rho,i\vDash(\phi_1\U\phi_2,L)$ if and only if there exists $j\ge i$ such that $\rho,j\vDash(\phi_2,L)$ and $\rho,k\vDash(\phi_1,L)$ for all $k\in [i,j)$;
\item $\rho,i\vDash(\phi_1\R\phi_2,L)$ if and only if, for all $j\ge i$, at least one of the following holds: $\rho,j\vDash(\phi_2,L)$ or there exists $k\in [i,j)$ such that $\rho,k\vDash(\phi_1,L)$.
\end{itemize}
Similarly, we write $\rho\vDash\phi$ if $\rho,0\vDash\phi$. If the labelling function is clear from the context, we simply write $\xi\vDash\phi$. 

\subsection{Proof of Proposition \ref{prop:intersample}}

\begin{proof}
The implementation of control strategy and preservation of dwell-time are straightforward. The main part is to show correctness of temporal logic formula. Suppose that a trajectory for $\S_\delta$ a control strategy is $\x(t)$. Let $\u$ be the resulting control input signal. The corresponding path of an execution of $\T_{\delta,\tau}$ is given by $\rho=\x(0),\x(\tau),\x(2\tau),\cdots$. 

We need to show that (1) $\rho\vDash (\phi,L)$ implies $x\vDash (\phi,L_\eps)$, and (2) $x\vDash (\phi,L_\eps)$ implies $\rho\vDash (\phi,L)$. The following proof, modelled after that for Theorem 4.1 in \cite{liu2016finite}, is an inductive argument based on the structure of $\ltlx$ formulas. In fact, the proof for (1) is very similar to of for Theorem 4.1 in \cite{liu2016finite}. In the following, we prove case (2), that is, $x\vDash (\phi,L_\eps)$ implies $\rho\vDash (\phi,L_\eps)$. We do so by proving a stronger statement: for every $i\ge 0$, $x,t\vDash (\phi,L_\eps)$ for some $t\in J_i=[i\tau-\frac{\tau}{2},i\tau+\frac{\tau}{2}]$ implies $\rho,i \vDash(\phi,L_\eps)$. 

%
%

\textbf{Case $\phi=\pi$}: Suppose that $x,t\vDash(\pi,L_\eps)$ for some $t\in J_i$, we have to show that $\pi\in L(x(i\tau))$. This follows from $\pi\in L_\eps(x(t))$, $\eps\ge (M+\delta)\tau/2$ and
\begin{equation}\label{eq:basic}
\abs{x(t)-x(i\tau)}\le\abs{x(t)-x(\tau_i)} \le (M+\delta)\tau/2.
\end{equation}


\textbf{Case $\phi=\phi_1\R\phi_2$}: Suppose that $x(t)\vDash(\phi,L_\eps)$ for some $t\in J_i$. We need to show that $\rho,i\vDash(\phi,L)$, that is, for all $j\ge i$, either $\rho,j\vDash(\phi_2,L)$ holds or there exists some $k\in [i,j)$ such that $\rho,k\vDash(\phi_1,L)$ holds. Since $x(t)\vDash(\phi,L_\eps)$ for some $t\in J_i$, we know that for every $t'\ge t$, either $x(t')\vDash(\phi,L_\eps)$ holds or there exists $s\in [t,t')$ such that $x(s)\vDash(\phi,L_\eps)$ holds. Let $t'=j\tau - \frac{\tau}{2}$. If the former holds, we have $x(t')\vDash(\phi,L_\eps)$ for $t'=j\tau - \frac{\tau}{2} \in J_j$. By the inductive assumption, this implies $\rho,j\vDash(\phi_2,L)$. If the latter holds, there exists some interval $J_k$ such that $s\in J_k$, $k\in [i,j)$, and $x(s)\vDash(\phi,L_\eps)$. It follows by the inductive assumption that $\rho,k\vDash(\phi_2,L)$. 


The other cases are straightforward.
\end{proof}

\subsection{Details of the nonlinear car model}

\begin{ex}[Nonlinear car]
Consider a nonlinear car with bicycle dynamics \cite{astrom2010feedback} (parameters are taken from  \cite{zamani2012symbolic}):  
  \begin{align*}
      {x}' &=  v \cos (\alpha+\theta)/\cos (\alpha),\\
      {y}' &= v \sin (\alpha+\theta)/\cos (\alpha),\\
      {\theta}' &= v \tan (\phi),
  \end{align*}
where $(x,y)$ are the position (centre of mass) and $\theta$ is the heading angle, the controls are $(v,\phi)$ with $v$ being the wheel speed and $\phi$ the steering angle. The wheel base is given by $b$ and $a$ is the distance between centre of mass and rear wheels. We choose $a=0.5$ and $b=1$ as in \cite{zamani2012symbolic}. The variable $\alpha=\arctan (a\tan (u_2)/b)$. 
\end{ex}

\end{document}